\documentclass[11pt,a4paper]{article}
\pdfoutput=1
\usepackage[left=2.3cm,right=2.3cm,top=2.3cm,bottom=2.3cm]{geometry}
\usepackage[ascii]{inputenc}
\usepackage{mathtools}
\usepackage{times}
\usepackage[dvipsnames,svgnames]{xcolor}
\usepackage{setspace}
\usepackage[bookmarksnumbered,linktocpage,hypertexnames=false,colorlinks=true,linkcolor=NavyBlue,urlcolor=NavyBlue,citecolor=ForestGreen,anchorcolor=green,breaklinks=true,pagebackref=true,pdfusetitle]{hyperref}
\usepackage[normalem]{ulem}
\usepackage{bbm,mathrsfs,amsmath,amssymb,amsthm,amsfonts,latexsym,graphicx,multirow,booktabs,bm,xspace,float,mathdots,caption,subcaption,ellipsis,mleftright,dsfont,faktor,kantlipsum,algorithm,tocloft,quantikz,tocloft}
\usepackage[shortlabels]{enumitem}
\usepackage[T1]{fontenc}
\usepackage{textcomp}
\usepackage[sort&compress,numbers]{natbib}
\usepackage[english]{babel}
\usepackage[capitalize,nameinlink]{cleveref}
\newcommand{\crefpart}[2]{\hyperref[#2]{\namecref{#1}~\labelcref*{#1}~\ref*{#2}}}
\newtheorem{theorem}{Theorem}[section]
\newtheorem{corollary}[theorem]{Corollary}\AddToHook{env/corollary/begin}{\crefalias{theorem}{corollary}}
\AddToHook{env/definition/begin}{\crefalias{theorem}{definition}}
\newtheorem{lemma}[theorem]{Lemma}\AddToHook{env/lemma/begin}{\crefalias{theorem}{lemma}}
\newtheorem{proposition}[theorem]{Proposition}\AddToHook{env/proposition/begin}{\crefalias{theorem}{proposition}}\crefname{proposition}{Proposition}{Propositions}
\AddToHook{env/conjecture/begin}{\crefalias{theorem}{conjecture}}
\AddToHook{env/remark/begin}{\crefalias{theorem}{remark}}
\newtheorem*{theorem*}{Theorem}
\AddToHook{cmd/appendix/after}{\crefalias{section}{appendix}}
\numberwithin{equation}{section}
\allowdisplaybreaks[4]
\newcommand{\ot}{\otimes}
\newcommand{\op}{\oplus}
\newcommand{\eps}{\varepsilon}
\newcommand{\dd}{\textrm{d}}
\newcommand{\bigO}{\mathcal O}
\renewcommand{\Re}{\operatorname{Re}}

\newcommand{\RR}{\mathbbm{R}}
\newcommand{\CC}{\mathbbm{C}}
\newcommand{\NN}{\mathbbm{N}}
\newcommand{\HH}{\mathcal{H}}

\newcommand{\id}{I}
\newcommand{\GG}{\mathcal G}

\newcommand{\1}{\mathbf{1}}
\let\oldS\S
\renewcommand{\S}{\operatorname{S}}
\renewcommand{\P}{\operatorname{P}}
\DeclareMathOperator{\dom}{dom}
\DeclareMathOperator{\grad}{grad}

\DeclareMathOperator{\tr}{tr}
\DeclareMathOperator{\Sp}{Sp}
\DeclareMathOperator{\SO}{SO}
\DeclareMathOperator{\Ortho}{O}
\DeclareMathOperator{\U}{U}
\DeclareMathOperator{\SU}{SU}
\DeclareMathOperator{\Spin}{Spin}

\DeclareMathOperator{\Sym}{Sym}

\DeclareMathOperator{\Span}{span}
\DeclareMathOperator{\spec}{spec}

\DeclarePairedDelimiter\abs{\lvert}{\rvert}
\DeclarePairedDelimiter\norm{\lVert}{\rVert}

\DeclarePairedDelimiter\ceil{\lceil}{\rceil}
\DeclarePairedDelimiter\parens{\lparen}{\rparen}
\DeclarePairedDelimiter\braces{\lbrace}{\rbrace}

\DeclarePairedDelimiter\angles{\langle}{\rangle}
\DeclarePairedDelimiter\ip{\langle}{\rangle}

\begin{document}

\title{Testing bosonic and fermionic Gaussian states with~a~constant~number~of~copies}
\author{Jop Bri\"{e}t\thanks{CWI \& QuSoft, Amsterdam, University of Leiden, \href{mailto:j.briet@cwi.nl}{j.briet@cwi.nl}}
\and
Davi Castro-Silva\thanks{Department of Computer Science, University of Warwick, \href{mailto:davi.de-castro-silva@warwick.ac.uk}{davi.de-castro-silva@warwick.ac.uk}}
\and
Michael Walter\thanks{LMU Munich \& MCQST and University of Amsterdam \& QuSoft, \href{mailto:michael.walter@lmu.de}{michael.walter@lmu.de}}
\and
Freek Witteveen\thanks{CWI \& QuSoft, Amsterdam, \href{mailto:f.witteveen@cwi.nl}{f.witteveen@cwi.nl}}}
\date{}
\maketitle
\begin{abstract}
Bosonic and fermionic Gaussian quantum states play a central role in quantum physics and quantum information processing.
We study the problem of testing whether a given pure state is Gaussian or is bounded away from the set of pure Gaussian states.
Earlier work left open whether a fixed number of copies could suffice even as the number of modes (system size) grows.
We show that it does: the number of copies needed depends only on the desired accuracy and confidence.
Our tests exploit the symmetries of identical copies of Gaussian states; their analysis combines a gradient flow argument with a spectral gap studied previously for the Kac random walk on the sphere.
\end{abstract}

\section{Introduction}
Gaussian quantum states are an important family of quantum states.
They appear as ground states of quadratic Hamiltonians in quantum optics and materials science, as a tractable set of quantum states for approximation or simulation by classical computation, and are often used in quantum communication protocols.
On the other hand, \emph{non-Gaussian} states also play a central role: they appear as ground states of more complicated Hamiltonians, and as a resource for quantum computation beyond classical simulability.
This motivates the following basic question: given copies of a state $\ket{\psi}$, can one perform a measurement that determines whether the state is Gaussian or is bounded away from the set of pure Gaussian states?
This question has been studied in previous works \cite{lyu2024fermionic,bu2025efficient,tarabunga2026fermionic,girardi2025gaussian}.
Earlier work gave exact characterizations using a fixed number of copies; however, the known bounds on the number of copies needed to achieve a desired success probability grew with the number of modes (see \cref{sec:related-work}).
Thus, the sample complexity of Gaussianity testing had so far remained open.

The present work shows that a \emph{constant} number of copies suffices to test Gaussianity with one-sided error for both fermionic and bosonic Gaussian states.%
\footnote{We provide a Lean formalization at \url{https://github.com/qit-lmu/gaussian-testing}.}

\begin{theorem}[Testing Gaussianity with constant sample complexity]\label{thm:main-intro}
    Both bosonic and fermionic systems admit tests that use $O(\frac{1}{\eps^2}\log\frac{1}{\delta})$ copies of a state $\ket{\psi}$ to achieve the following:
    if $\ket\psi$ is Gaussian, the test always accepts,
    while if $\abs{\braket{\phi}{\psi}}^2 \leq 1 - \eps^2$ for all Gaussian states~$\ket{\phi}$, it rejects with probability~at~least~$1-\delta$.
\end{theorem}

\noindent
Note that $\eps$ is the trace-distance threshold, since for pure states $\tfrac12\norm{\proj{\psi}-\proj{\phi}}_1=\sqrt{1-\abs{\braket{\phi}{\psi}}^2}$.

We also give tests for mean-zero bosonic Gaussian states, and for coherent states.
The tests act on two copies at a time when testing fermionic Gaussian states, mean-zero bosonic Gaussian states, or coherent states.
The test for bosonic Gaussian states uses four copies at a time.

Our tests use well-known characterizations of Gaussian states by their invariance under phase-space rotations.
For fermionic Gaussian states and bosonic mean-zero Gaussian states, we consider a two-copy test with accepting operator $P$, which projects onto states which are invariant under rotations (beam-splitters of arbitrary angles) between the two copies.
These measurements have explicit implementations: \cite{haug2026practical,tarabunga2026fermionic} give a measurement scheme in terms of Bell sampling, and \cite{hertz2019multicopy} gives an implementation in the bosonic case.
For the case of non-zero mean, we apply a balanced beam splitter to two copies and retain the difference output. This output is even under parity, so its mean vanishes whenever its first moments exist. This can be used for testing coherent states and for reducing the general bosonic Gaussianity test to the mean-zero test. See \cref{fig:tests} for an overview.
These tests all have perfect completeness.
In fact, they accept a state with probability one if, and only if, it is in the target class.

The challenging part in bounding the sample complexity consists in analyzing the two-copy tests.
Here we use a tool that may be of wider interest, which we briefly explain in a more general context.
Assume that we have a two-copy test with accepting projection operator $P$, and write $p(\psi) = \bra{\psi}^{\ot 2} P \ket{\psi}^{\ot 2}$ for the probability of accepting a state.
Assume $\GG$ consists exactly of the states with $p(\psi) = 1$ and is nonempty.
What we need to show is that, if $p(\psi)$ is high, then $\ket{\psi}$ is close to $\GG$.
We measure closeness to the set~$\GG$ using the standard notion of fidelity $F_\mathcal G(\psi) = \sup_{\phi\in \mathcal G} \abs{\braket{\phi}{\psi}}$.
This means that if a state is far from~$\GG$, the test rejects it with probability bounded below.
We show that this can be related to the spectrum of a 3-copy operator.
Below, we write $P_{ab}$ for $P$ acting on copies $a$ and $b$.

\begin{theorem}[Robust property testing from spectral gap, informal]\label{thm:gradient-flow-intro}
Suppose $P$ commutes with the swap operator and let $S = \frac13\left(P_{12} + P_{13} + P_{23}\right)$.
Suppose that $\spec(S) \subseteq [0,\alpha] \cup \{1\}$ for some $\alpha\in[0,2/3)$. Then there is a constant $C_\alpha$ such that, for every $\tau\geq0$,
\begin{align*}
    p(\psi) \geq 1 - \tau \quad\Rightarrow\quad F_\GG^2(\psi) \geq 1 - C_\alpha \tau.
\end{align*}
\end{theorem}

\noindent
It turns out that in the case of testing Gaussianity, the natural choice for~$P$ leads to an operator~$S$ that is directly related to the transition operator of Kac's random walk on the sphere~\cite{kac1956foundations}, whose spectral gap is known~\cite{carlen2003determination}.

\begin{figure*}[t]
\centering

\newcommand{\testnote}[1]{%
  \par\vspace{0.4em}%
  {\scriptsize
    \begin{minipage}[t][2\baselineskip][t]{\linewidth}
      \centering #1\par
    \end{minipage}\par}%
}


\begin{subfigure}[t]{0.47\linewidth}
\vspace{0pt}
\centering

\begin{minipage}[c][2.2cm][c]{\linewidth}
\centering
\begin{quantikz}[row sep=0.35cm,column sep=0.35cm]
\lstick{$\ket{\psi}$}
  & \gate[2]{U_{\mathrm{Bell}}}
  & \meter{}
  & \setwiretype{c} \\
\lstick{$\ket{\psi}$}
  &
  & \meter{}
  & \setwiretype{c}
\end{quantikz}
\end{minipage}

\testnote{Accept if the Bell outcome is compatible with $\Lambda=0$.}

\caption{Fermionic Gaussianity test.}
\label{fig:test-fermions}
\end{subfigure}%
\hfill
\begin{subfigure}[t]{0.47\linewidth}
\vspace{0pt}
\centering

\begin{minipage}[c][2.2cm][c]{\linewidth}
\centering
\begin{quantikz}[row sep=0.5cm,column sep=0.4cm]
\lstick{$\ket{\psi}$}
  & \qw
  & \gate[2]{U_{\pi/4}}
  & \meter{N_1}
  & \setwiretype{c} \\
\lstick{$\ket{\psi}$}
  & \gate{R_{\pi/2}}
  &
  & \meter{N_2}
  & \setwiretype{c}
\end{quantikz}
\end{minipage}

\testnote{Shift every mode of the second register by $\pi/2$. After a balanced beam splitter, accept if the total counts agree ($N_1=N_2$).}


\caption{Mean-zero bosonic Gaussianity test.}
\label{fig:test-boson-zero-mean}
\end{subfigure}

\par\vspace{1.2em}


\begin{subfigure}[t]{0.47\linewidth}
\vspace{0pt}
\centering

\begin{minipage}[c][3.6cm][c]{\linewidth}
\centering
\begin{quantikz}[row sep=0.35cm,column sep=0.35cm]
\lstick{$\ket{\psi}$}
  & \gate[2]{U_{-\pi/4}}
  & \qw
  & \setwiretype{n} \\
\lstick{$\ket{\psi}$}
  &
  & \meter{N}
  & \setwiretype{c}
\end{quantikz}
\end{minipage}

\testnote{Apply a beam splitter and accept if total particle number $N = 0$.}

\caption{Coherent-state test.}
\label{fig:test-coherent}
\end{subfigure}%
\hfill
\begin{subfigure}[t]{0.47\linewidth}
\vspace{0pt}
\centering

\begin{minipage}[c][3.6cm][c]{\linewidth}
\centering
\begin{quantikz}[row sep=0.24cm,column sep=0.32cm]
\lstick{$\ket{\psi}$}
  & \gate[2]{U_{-\pi/4}}
  & \qw
  & \setwiretype{n}
  & \\
\lstick{$\ket{\psi}$}
  &
  & \qw
  & \gate[3]{\text{\small as in (b)}}
  & \setwiretype{c} \\
\lstick{$\ket{\psi}$}
  & \gate[2]{U_{-\pi/4}}
  & \qw
  & \setwiretype{n}
  & \\
\lstick{$\ket{\psi}$}
  &
  & \qw
  &
  & \setwiretype{c}
\end{quantikz}
\end{minipage}

\testnote{Retain the difference output from each balanced beam splitter, and apply the test from panel~(b) to these two outputs.}

\caption{General bosonic Gaussianity test.}
\label{fig:test-boson-general}
\end{subfigure}

\caption{
Implementations of the four tests considered in this work. Each wire represents an $m$-mode Hilbert space.
(a) For fermions, a two-copy Bell measurement is performed and its classical outcome is post-processed. See \cite{haug2026practical,tarabunga2026fermionic} for details.
(b) For mean-zero bosons, one may realize the two-copy observable by a $\pi/2$ phase shift on every mode of the second register, a balanced beam splitter, and measuring the particle-number difference \cite{hertz2019multicopy}.
(c) For coherent states, a balanced beam splitter is applied and one output is tested to be in the vacuum state.
(d) For general bosonic Gaussian states, the test from panel~(b) acts on the difference outputs of two balanced beam splitters.
}
\label{fig:tests}
\end{figure*}
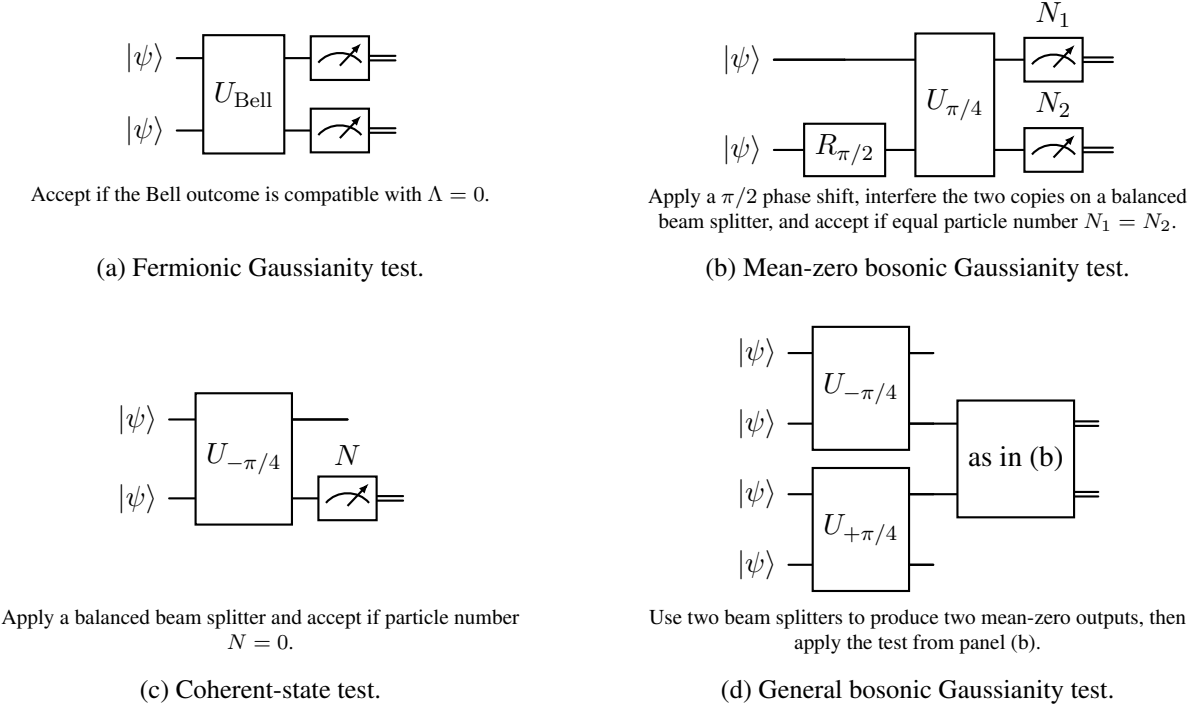

\subsection{Techniques}
We consider a system of $m$ bosonic or fermionic modes, with creation and annihilation operators $a_j^\dagger, a_j$ for $j = 1,\dots,m$. A state $\ket{\psi}$ is a pure Gaussian state if it is the nondegenerate ground state of a Hamiltonian that is quadratic in the creation and annihilation operators (where in the bosonic case linear terms are allowed).
A Gaussian state is fully determined by its first and second moments in the creation and annihilation operators.

In order to test whether a state is Gaussian, we need some conditions that single out pure Gaussian states among other states.
One such condition can be derived from the covariance matrix. In the bosonic case, a necessary and sufficient condition is that the symplectic eigenvalues of the covariance matrix are all equal to~1; in the fermionic case the singular values of the (antisymmetric) covariance matrix must be all equal to~1.
Here the bosonic covariance is the expectation of the anticommutator of centered quadratures (twice the convention in which the vacuum covariance is $I/2$); this criterion presupposes finite second moments. Thus, one natural approach to testing Gaussianity would be to estimate the covariance matrix from measurements, and use this to decide whether the state is Gaussian or not \cite{bittel2025optimal,girardi2025gaussian}.
In this work, we start from an alternative condition. For simplicity, we restrict the discussion here in the bosonic case to Gaussian states with zero mean, and refer to these simply as \emph{Gaussian states}.

\begin{description}
\item \textbf{Fermions:}
Let $\Lambda = 2 \sum_{j=1}^m \left(a_j \ot a_j^\dagger + a_j^\dagger \ot a_j\right)$.
Then $\ket{\psi}$ is Gaussian if and only if $\Lambda \ket{\psi}^{\ot 2} = 0$.
This characterization is due to~\cite{bravyi2005lagrangian,melo2013power}.

\item \textbf{Bosons:}
Let $\Lambda = i\sum_{j=1}^m \left(a_j \ot a_j^\dagger - a_j^\dagger \ot a_j\right)$.
Then $\ket{\psi}$ is a (mean-zero) Gaussian state if and only if~$\Lambda \ket{\psi}^{\ot 2} = 0$. This condition is also well-known, see e.g.~\cite{girardi2025gaussian} for a recent discussion.
\end{description}

\noindent
In the bosonic case, $\Lambda$ is a generator of phase space rotations between the two copies, and in the fermionic case it is closely related to a generator of phase space rotations (up to the difference between the fermionic and normal tensor product).
This is closely related to the symmetry properties of Gaussian states, see \cite{girardi2025gaussian,chen2026optimal}.
The details of this are not necessary for this work, but we briefly summarize them for context.
In Schur--Weyl duality, $U \in \U(d)$ acts on $(\CC^d)^{\ot n}$ as $U^{\ot n}$, and this action has a commutant spanned by the action of the symmetric group $S_n$ acting as permutations between the copies.
If we restrict $U$ to parity-preserving fermionic Gaussian unitaries or to symplectic bosonic Gaussian unitaries (without displacements), the commutant contains \emph{additional symmetries}.
The commutant turns out to be generated by the natural action of the orthogonal group~$\Ortho(n)$ on the many-copy phase space (instead of just the permutation action of the symmetric group).%
\footnote{If one allows displacements, the commutant restricts to the \emph{stochastic} orthogonal transformations, which fix the~all-ones~vector.}
This is known as \emph{Howe--Kashiwara--Vergne} duality~\cite{howe1989remarks,kashiwara1978segal}.

In this light, measuring rotation-invariance is a natural test for Gaussianity, and we prove that this test indeed has good robustness properties, independent of the number of modes.
In more detail, we use a test where we measure whether $\ket{\psi}^{\ot 2}$ is in the kernel of $\Lambda$.
In fact, the kernel of $\Lambda$ (intersected with the symmetric subspace in the bosonic case) is precisely the (closed) linear span of states~$\ket{\psi}^{\ot 2}$, for (zero-mean) Gaussian states.
This is analogous to the fact that the image of the projection on the symmetric subspace is spanned by states $\ket{\psi}^{\ot n}$ for arbitrary $\ket{\psi}$.
This means that we get a test with perfect completeness (it never rejects any Gaussian state), and it is the optimal such test on two copies (since it measures exactly the span of states $\ket{\psi}^{\ot 2}$ for Gaussian states).
In order to prove \cref{thm:main-intro}, what we need to show is a \emph{robustness property}: if a state is far from Gaussian, then it is rejected with probability bounded below in terms of the distance.
Or, in other words, if the acceptance probability $p(\psi)$ on input $\ket{\psi}$ is high, $p(\psi) \geq 1 - \tau$, then the distance to a Gaussian state must be small in a way that does not depend on the number of modes $m$.
Such a bound follows directly from \cref{thm:gradient-flow-intro}, if we bound the spectral gap of the operator $S = \frac13\left(P_{12} + P_{13} + P_{23}\right)$.
In the case of Gaussian states, the operators $P_{ab}$ are related to phase space rotations in three different directions, and are related to an action of $\SO(3)$ on the three-copy space.
It turns out that $S$ is closely related to a different operator, namely the transition operator of the \emph{Kac walk} \cite{kac1956foundations}, a random walk on the unit sphere of~$\RR^N$, where for us $N=3$.
For that operator, bounds on the spectral gap are known \cite{carlen2000many,diaconis2000bounds,maslen2003eigenvalues,carlen2003determination,caputo2008spectral} and we use these to bound the spectral gap of $S$.

We now sketch the proof of \cref{thm:gradient-flow-intro}.
Recall that we have a test with accepting projection operator~$P$, write $p(\psi) = \bra{\psi}^{\ot 2} P \ket{\psi}^{\ot 2}$ for the probability of accepting a state, and we are testing states in~$\GG$, which are exactly those with $p(\psi) = 1$.
The idea is to follow the gradient flow for the acceptance probability~$p$ to a point where the gradient is zero.
When $p$ is sufficiently close to one, we prove a lower bound on its gradient proportional to $\sqrt{1-p}$.
If we start the flow in this regime, this bound implies that $p$ approaches one and it also allows us to control the length of the path to its limit, which must be in~$\GG$.
In optimization, gradient bounds of this kind are known as \emph{\L{}ojasiewicz inequalities}~\cite{lojasiewicz1965ensembles};
here a strong quantitative version is needed to establish quantitative bounds.
To establish the lower bound on the gradient of~$p$, it suffices to bound the first and second moments of~$S$ using the assumed spectral gap.

\subsection{Prior work}\label{sec:related-work}
Due to the central role played by Gaussian states in physics and quantum information processing, they have been extensively studied in quantum learning theory and quantum information theory. There is also the class of \emph{stabilizer states}, which occur frequently in quantum computation and quantum error correction, and which can be seen as a discrete version of bosonic Gaussian states. It turns out that all of these classes can be tested using a constant number of copies, and the underlying logic is similar between the different classes.

Indeed, for stabilizer states it is shown in \cite{gross2021schur} that one can test stabilizer states using a constant number of copies.
The reason to mention this result is that this can be done by exploiting the symmetry of the problem, in a way closely analogous to this work. Clifford unitaries act as phase space transformations on a discrete phase space, and $U^{\ot n}$ for Clifford $U$ has a commutant which in particular contains rotations between the copies of phase space. Measuring this additional symmetry yields a constant-copy test for stabilizer states. Further works \cite{arunachalam2025polynomial,bao2025tolerant, MehrabanTahmasbi:2025} have proven stronger results, with different techniques based on additive combinatorics, showing that it is also possible to test whether a state has small or large maximal overlap with a stabilizer state.

The Gaussianity condition analyzed in this work and conditions similar to it are folklore and have been observed before.
For the bosonic case, these can be found, e.g., in \cite{mann1991bifactorizable,springer2009conditions}, and were investigated more recently from the perspective of a testing algorithm in \cite{bu2025efficient} and \cite{girardi2025gaussian}. Some robustness properties were shown in \cite{cuesta2020stable} (but these error bounds scale with the number of modes in a way that does not lead to a satisfying analysis of the test).
The previous best sample complexity bounds for testing bosonic Gaussianity are based on learning the Gaussian state (e.g., combining the results in \cite{girardi2025gaussian} with optimal tomography of Gaussian states \cite{chen2026optimal} would give a scaling of $O(m^2 \eps^{-2})$). A dimension-independent robustness result was already established in \cite{girardi2025gaussian}, but only for states close to Gaussian states; our results here extend this result to a global bound.
In the fermionic case, the operator $\Lambda$ was used in the context of quantum information processing in \cite{bravyi2005lagrangian,melo2013power} to characterize Gaussian states.
By now, there is an extensive body of work on testing fermionic non-Gaussian states, and quantifying non-Gaussianity more generally
\cite{lyu2024fermionic,bittel2025optimal,coffman2025measuring,leone2026fermionic,tarabunga2026fermionic,haug2026practical,sierant2026fermionic,tarabunga2026computable}, and several of these works use the same symmetry principles; in particular, \cite{haug2026practical,tarabunga2026fermionic} study the properties of $\Lambda$ as a fermionic Gaussianity test and prove robustness bounds (that still scale with $m$).
The best previously known bound for the sample complexity to test Gaussianity is $\tilde O(m \eps^{-2})$ \cite{leone2026fermionic}.

\paragraph{Concurrent work}
Closely related work has been done independently by Iyer, O'Gorman, Parekh, Thompson, and Zhao~\cite{iyer2026optimal}, by Yaghubi~Rad, Puig, Hassanandani, Coffman, Eisert, Bravo-Prieto, and Mele~\cite{rad2026testing}, and by West and Larocca \cite{west2026fermionic}.

\section{Preliminaries}
Let $\HH$ be a (possibly infinite-dimensional) Hilbert space.
We denote by $\S(\HH)$ the density operators on $\HH$, i.e.\ the positive semidefinite trace-class operators with trace one, and by $\P(\HH) \subseteq \S(\HH)$ the pure states.
We denote pure states by $\phi = \proj\phi$, where $\ket\phi$ is the corresponding unit vector, unique up to an overall phase, and vice versa.
The expectation value of a bounded operator~$O$ in a state~$\rho$ is denoted by $\langle O \rangle_\rho = \tr(\rho O)$, so $\langle O \rangle_\phi = \ip{\phi | O | \phi}$ for a pure state~$\phi$.
All projections are orthogonal projections.

\section{Robustness of property tests from gradient flows}\label{sec:gradient-flow}
In this section we discuss a general method for establishing the robustness of perfectly complete tests.
The goal is to show that if a state is accepted with high probability, then it is near a state that is accepted with probability one.
The key idea is to analyze the gradient flow of the acceptance probability and bound the gradient in terms of the rejection probability (such a bound is also known as a \L{}ojasiewicz inequality).
The result of the analysis is summarized in \cref{thm:bound-fidelity-gradient-flow}.

In this section, let $\HH$ be an arbitrary (possibly infinite-dimensional) Hilbert space.
Let $P$ be any orthogonal projection on $\HH^{\ot 2}$ that commutes with the swap operator between the two copies.
Such a test determines a binary projective measurement on two copies of a pure quantum state, with acceptance probability
\begin{align}\label{eq:p-accept-2-copies}
    p \colon \P(\HH) \to [0,1], \quad p(\psi) \coloneq \ip{P}_{\psi^{\ot 2}} = \ip{\psi^{\ot 2} | P | \psi^{\ot 2}}.
\end{align}
Denote by $\GG$ the set of (``good'') pure states that are accepted with probability one, equivalently $P \psi^{\ot 2} = \psi^{\ot 2}$, and by $F_{\GG}(\psi)$ the maximal fidelity of a given state with one in~$\GG$:
\begin{align}\label{eq:define-F-and-G}
    \GG \coloneqq \{ \psi \in \P(\HH) \, : \, p(\psi) = 1 \} \, ,
    \qquad
    F_{\GG}(\psi) \coloneqq \smash{\sup_{\phi \in \GG}} \, \abs{\ip{\psi|\phi}},
\end{align}
where we assume that $\GG\neq\emptyset$.

We consider the \emph{gradient flow} for the acceptance probability function.
This is the following ordinary differential equation in $\P(\HH)$:
\begin{align}\label{eq:gradient flow abstract}
    \dot \psi_t = \grad p(\psi_t) \qquad (t \geq 0)
\end{align}
starting from some initial state $\psi_0 \in \P(\HH)$.
In finite dimensions, the projective space of pure states~$\P(\HH)$ is a finite-dimensional Riemannian manifold (in fact, a complex K\"ahler manifold):
the tangent space at any point $\psi = \proj\psi$ is naturally identified with the orthogonal complement, $T_\psi \P(\HH) \cong \ket\psi^\perp \subseteq \HH$; then the Riemannian metric is given by $g_\psi(\alpha,\beta) = \Re \braket\alpha\beta$; we write $\norm{\cdot}$ for the corresponding local norm since it is simply the Hilbert space norm in these coordinates.
The corresponding distance on~$\P(\HH)$ is known as the Fubini--Study distance~$d_\mathrm{FS}(\psi,\phi) \coloneqq \arccos \, \abs{\ip{\psi|\phi}}$.
The gradient of the acceptance probability is readily computed using the product rule and the fact that $P$ commutes with the swap operator:
\begin{align}\label{eq:acceptance grad}
    \grad p(\psi)
= 4 (I \ot \bra\psi) P \ket\psi^{\ot 2} - 4 \, p(\psi) \ket\psi
\qquad (\psi=\proj\psi \in \P(\HH)).
\end{align}
The first term would be the gradient of the function~\eqref{eq:p-accept-2-copies} if it were defined on the entire Hilbert space; the second term achieves the projection onto the tangent space of $\P(\HH)$.
In infinite dimensions, $\P(\HH)$ is still a Riemannian Hilbert manifold, resulting in the same formula for the gradient of~$p$; alternatively, one can simply take \cref{eq:acceptance grad} as the definition of the gradient vector field.

\begin{lemma}\label{lem:ode}
The gradient flow~\eqref{eq:gradient flow abstract} has a smooth unique global solution $(\psi_t)_{t\geq0}$ for any initial state~$\psi_0 \in \P(\HH)$.
\end{lemma}
\begin{proof}
The standard local existence and uniqueness theorem for ODEs given by smooth vector fields (such as $\grad p$) due to Picard--Lindel\"of extends to smooth infinite-dimensional manifolds modeled on Banach spaces (see \cite[\oldS{}1.10]{moore2017introduction} or \cite[\oldS{}IV]{lang2012fundamentals}).
By gluing together local solutions, we obtain a smooth unique maximal solution $(\psi_t)_{t\in[0,T)}$ for some $T \in (0,\infty]$.
We claim that $T = \infty$.
Indeed, suppose that $T<\infty$.
Then, for any $0 \leq t_0 \leq t_1 < T$, we have
\begin{align*}
    d_{\mathrm{FS}}(\psi_{t_0}, \psi_{t_1})^2
&\leq \left( \int_{t_0}^{t_1} \norm{\dot\psi_t} \dd t \right)^2
\leq (t_1 - t_0) \int_{t_0}^{t_1} \norm{\dot\psi_t}^2 \, \dd t
= (t_1 - t_0) \int_{t_0}^{t_1} \frac{\dd p(\psi_t)}{\dd t} \, \dd t \\
&= (t_1 - t_0) \left( p(\psi_{t_1}) - p(\psi_{t_0}) \right)
\leq t_1 - t_0,
\end{align*}
where we first used the fact that the distance is upper bounded by the length of any curve, then the Cauchy--Schwarz inequality,
next that $\frac{\dd p(\psi_t)}{\dd t} = \norm{\grad p(\psi_t)}^2 = \norm{\dot\psi_t}^2$, which holds for any gradient flow directly from the definition~\eqref{eq:gradient flow abstract} and the chain rule, and finally that $0 \leq p \leq 1$.
Thus, $\psi_t$ is Cauchy as $t\uparrow T$.
Because a Hilbert space and hence also~$\P(\HH)$ is complete, the solution can be continued to $t=T$ and, by local existence and uniqueness, to a neighborhood of~$T$.
This contradicts the maximality of $T$.
\end{proof}

Note that the acceptance probability and its gradient can be written as
\begin{align*}
    p(\psi) = \ip{\psi | K_\psi | \psi} \quad\text{and}\quad \grad p(\psi) = 4 \left( K_\psi - p(\psi) I \right) \ket\psi,
\end{align*}
using \cref{eq:acceptance grad},
where $K_\psi$ is the positive semidefinite contraction on $\HH$ defined by
\begin{align*}
    K_\psi \coloneqq (I \ot \bra\psi) P (I \ot \ket\psi).
\end{align*}
In particular, the squared norm of the gradient at any point $\psi \in \P(\HH)$ has a pleasant interpretation in terms of the variance of this operator:
we have $\norm{\grad p(\psi)}^2 = 16 ( \ip{K_\psi^2}_\psi - \ip{K_\psi}_\psi^2 )$ or
\begin{align}\label{eq:grad via p and q}
    \norm{\grad p}^2 = 16 \left( q - p^2 \right),
    \quad\text{where}\quad
    p(\psi) = \ip{\psi | K_\psi | \psi} \text{ and }
    q(\psi) \coloneqq \ip{\psi | K_\psi^2 | \psi}.
\end{align}
In particular, $\psi$ is a critical point of the acceptance probability if and only if $K_\psi \!\ket{\psi} = p(\psi)\ket{\psi}$.
Unfortunately the operator $K_\psi$ depends on $\psi$, which makes it difficult to reason about directly.

Instead, we use a symmetrization trick and consider the following state-independent operator on three copies of the Hilbert space:
\begin{equation}\label{eq:def-S}
    S \coloneqq \frac13 \left( P_{12} + P_{13} + P_{23} \right),
\end{equation}
on $\HH^{\ot 3}$, where $P_{ab}$ acts as $P$ on copies $a$ and $b$ (the order of indices does not matter because $P$ commutes with the swap operator) and as the identity on the third system.
The first and second moments of $S$ encode the first and second moments of $K_\psi$ in a transparent way.

\begin{lemma}\label{lem:three-copy-moments}
Let $p, q \colon \P(\HH) \to [0,1]$ be defined as in \cref{eq:p-accept-2-copies,eq:grad via p and q}.
Then, for every $\psi \in \P(\HH)$, we have
\begin{align*}
    \ip{S}_{\psi^{\ot 3}} = p(\psi)
    \quad\text{and}\quad
    \ip{S^2}_{\psi^{\ot 3}} = \frac13\left( p(\psi) + 2q(\psi) \right).
\end{align*}
\end{lemma}
\begin{proof}
    It is clear that $\ip{S}_{\psi^{\ot 3}} = p(\psi)$ because we have $\ip{P_{ab}}_{\psi^{\ot 3}} = p(\psi)$ for all $a\neq b$.
    Next, if we expand~$\ip{S^2}_{\psi^{\ot 3}}$, there are three terms corresponding to $\frac19 P_{ab}^2 = \frac19 P_{ab}$; their sum results in the term $p(\psi)/3$.
    For the cross terms, we have
    \begin{align*}
        \ip{\psi^{\ot 3} | P_{12} P_{23} | \psi^{\ot 3}}
    = \ip{P \psi^{\ot 2} \ot \psi \,|\, \psi \ot P \psi^{\ot 2}}
    = \norm{(I \ot \bra\psi) P \ket\psi^{\ot 2}}^2
    = \norm{K_\psi \ket\psi}^2
    = q(\psi)
    \end{align*}
    using the permutation symmetry of $P$ in the second step.
    Because $\psi^{\ot 3}$ is in the symmetric subspace, all six cross terms are equal to $1/9$ times the above.
\end{proof}

It is clear that $S$ is a positive semidefinite contraction, hence its spectrum is in~$[0,1]$.
Moreover, $S \ket\psi^{\ot 3} = \ket\psi^{\ot 3}$ if and only if~$\psi \in \GG$; thus its maximal eigenvalue is~$1$.
We now show that if $S$ has a good spectral gap below this maximal eigenvalue then the gradient can be lower bounded in terms of the success probability.
The following bound can be understood as an effective form of a \emph{\L{}ojasiewicz inequality}~\cite{lojasiewicz1965ensembles}.%
\footnote{Since $S$ commutes with permutations and its moments are only evaluated on~$\ket\psi^{\ot 3}$, for \cref{lem:bound-gradient} it in fact suffices that the spectral gap hypothesis holds only on $\Sym^3(\HH)$.}

\begin{lemma}\label{lem:bound-gradient}
    Suppose that $\spec(S) \subseteq [0,\alpha] \cup \{1\}$ for some $\alpha\in[0,1]$.
    Then:
    \begin{align*}
        \norm{\grad p}^2 \geq 16 \left( p - \tfrac32\alpha \right) \left( 1 - p \right).
    \end{align*}
\end{lemma}
\begin{proof}
Because $\spec(S) \subseteq [0,1]$, the assumption on the spectrum of $S$ is equivalent to $(S - I)(S - \alpha) \geq 0$.
Thus, we have, for every $\psi\in\P(\HH)$,
\begin{align*}
    0
\leq \ip{ (S - I)(S - \alpha) }_{\psi^{\ot 3}}
= \frac13\left( p(\psi) + 2q(\psi) \right) - (1 + \alpha) p(\psi) + \alpha,
\end{align*}
using \cref{lem:three-copy-moments}.
Rearranging the terms gives $q \geq p - \frac32\alpha (1 - p)$; the claim follows from \cref{eq:grad via p and q}.
\end{proof}

If $1 > p > 3\alpha/2$ (so in particular $\alpha < 2/3$), then the above bound is nontrivial.
We may use it to bound the distance to~$\GG$ in terms of the acceptance probability of the test.

\begin{theorem}[Robust property testing from spectral gap]\label{thm:bound-fidelity-gradient-flow}
Let $P$ be an orthogonal projection on $\HH^{\ot 2}$ that commutes with the swap operator, and set $S \coloneqq \frac13 ( P_{12} + P_{13} + P_{23} )$.
Suppose $\spec(S) \subseteq [0,\alpha] \cup \{1\}$ for some $\alpha\in[0,2/3)$, and set $C_\alpha \coloneqq 1/(1 - \tfrac32\alpha)$.
Then, for all $\psi\in\P(\HH)$ with~$p(\psi) \coloneqq \ip{\psi^{\ot2} | P | \psi^{\ot 2}} > \tfrac32\alpha$, we have
\begin{align*}
    \arccos F_\GG(\psi) \leq \frac12 \arcsin\sqrt{C_\alpha \left( 1-p(\psi) \right)}.
\end{align*}
As a consequence, the following holds for all $\psi\in\P(\HH)$ and $\tau\geq0$:
\begin{align*}
    p(\psi) \geq 1 - \tau \quad\Rightarrow\quad F_\GG^2(\psi) \geq 1 - C_\alpha \tau.
\end{align*}
Moreover, 
for each $\psi$ there exists $\phi\in\GG$ such that the stated bounds hold with $\abs{\braket\phi\psi}$ in place of $F_\GG(\psi)$.
\end{theorem}
\begin{proof}
Let $\psi\in\P(\HH)$.
Let $(\psi_t)_{t\geq0}$ denote the unique global solution to the gradient flow~\eqref{eq:gradient flow abstract} with initial state $\psi_0=\psi$ (\cref{lem:ode}).
Let us abbreviate by $p_t \coloneq p(\psi_t)$ the value of the objective along the flow.
If~$p_0 = 1$, then $\psi \in \GG$ and the theorem is immediate.
If $p_0 \leq \frac32\alpha$, then the hypothesis of the first inequality is not satisfied and the second one is vacuous.
Thus, we assume that $\frac32\alpha < p_0 < 1$.
Then, using \cref{lem:bound-gradient},
\begin{align}\label{eq:dot p t}
\dot p_t
= \norm{\dot \psi_t}^2
= \norm{\grad p(\psi_t)}^2
\geq 16 \left( p_t - \tfrac32\alpha \right) \left( 1 - p_t \right)
\geq 16 \left( p_0 - \tfrac32\alpha \right) \left( 1 - p_t \right).
\end{align}
The first equality shows that the objective~$p_t$ is non-decreasing (a general property of the gradient flow already used above), which in turn implies the final inequality.

Next, we claim that $p_t$ is \emph{strictly} increasing from $p_0$ at time~$t=0$ to $\lim_{t\to\infty} p_t = 1$.
The latter follows because $p_t$ is bounded by one and $p_0 > \tfrac32\alpha$, hence the right-hand side of~\eqref{eq:dot p t} must converge to zero.
The former also follows from~\eqref{eq:dot p t} if we can show that $p_t < 1$ for all $t \geq 0$.
Indeed, $q \leq p$ and hence, by \cref{eq:grad via p and q}, $\dot p_t = \norm{\grad p(\psi_t)}^2 \leq 16 (p_t - p_t^2) \leq 16 ( 1 - p_t )$, hence $1 - p_t \geq e^{-16t} (1 - p_0) > 0$.
This proves the claim.

Next we show that $\lim_{t\to\infty} \psi_t$ exists as well.
It suffices to argue that $(\psi_t)_{t\geq0}$ is Cauchy as $t\to\infty$.
Indeed, for any two $t_0 \leq t_1$, we have
\begin{align*}
    d_{\mathrm{FS}}(\psi_{t_0}, \psi_{t_1})
&\leq \int_{t_0}^{t_1} \norm{\dot\psi_t} \, \dd t
= \int_{t_0}^{t_1} \norm{\grad p(\psi_t)} \, \dd t
= \int_{p_{t_0}}^{p_{t_1}} \frac{\dd p}{\norm{\grad p(\psi_{t(p)})}} \\
&\leq \frac14 \int_{p_{t_0}}^{p_{t_1}} \frac{\dd p}{\sqrt{(p - \frac32\alpha)(1 - p)}}
\leq \frac14 \int_{p_{t_0}}^1 \frac{\dd p}{\sqrt{(p - \frac32\alpha)(1 - p)}}
= \frac12 \arcsin\sqrt{\frac{1-p_{t_0}}{1-\frac32\alpha}},
\end{align*}
where we first bound the distance by the length of the gradient flow trajectory, then we perform the change of variables $[0,\infty) \ni t \mapsto p_t \in [p_0,1)$, which by the above discussion is strictly increasing, with $\dd p = \dot p_t \, \dd t = \norm{\grad p(\psi_t)}^2 \, \dd t$, next we use \cref{eq:dot p t} to upper bound the integrand, and finally we evaluate the integral and use that $p_{t_1} \leq 1$.
Because $\lim_{t\to\infty} p_t=1$, the above shows that $(\psi_t)$ is Cauchy as $t\to\infty$.
Thus, $\psi_\infty \coloneqq \lim_{t\to\infty} \psi_t$ exists, satisfies $p(\psi_\infty)=1$ and hence $\psi_\infty \in \GG$, and the above computation (with~$t_0=0$, so $p_0 = p(\psi)$, and $t_1=\infty$) shows that
\begin{align*}
\arccos F_{\GG}(\psi)
\leq d_{\mathrm{FS}}(\psi, \psi_\infty)
\leq \int_0^\infty \norm{\dot\psi_t} \, \dd t
\leq \frac12 \arcsin\sqrt{\frac{1-p(\psi)}{1-\frac32\alpha}}.
\end{align*}
This proves the first estimate.
The second estimate is trivial for $p(\psi) \leq \frac32\alpha$, and otherwise follows from the first because $\cos^2(\frac12\arcsin\sqrt{x}) = \frac12 (1 + \cos\arcsin\sqrt{x}) = \frac12 (1 + \sqrt{1 - x}) \geq 1 - x$ for $x\in[0,1]$.
\end{proof}

\section{Spectral gap for the Kac walk}\label{sec:kac}
While $t$ copies of any pure state have an $S_t$-permutation symmetry, Gaussian states are characterized by having an enhanced symmetry, given by suitable groups of orthogonal transformations in~$O(t)$.
This will be explained in more detail in \cref{sec:fermions,sec:bosons}.
Here we discuss the common algebraic structure and explain how the relevant spectral gap can be extracted from the literature on the so-called \emph{Kac random walk}~\cite{kac1956foundations}.

The situation both for fermionic and mean-zero bosonic Gaussian states is as follows.
The group $\SO(t)$ acts on the $t$-copy Hilbert space.
Gaussian states are characterized by their invariance already given~$t=2$ copies, so one obtains a perfectly complete test by defining $P$~as the projection onto the $\SO(2)$-invariant subspace.
How can we obtain the symmetrized operator $S$ as in \cref{eq:def-S}, whose spectral gap determines the robustness of the test?
For this, one may consider the action of~$\SO(3)$, which acts on the three-copy Hilbert space and has three natural subgroups isomorphic to $\SO(2)$, corresponding to rotations about the $x$-, $y$-, and $z$-axes, respectively;
restricting the action to these subgroups corresponds precisely to the $\SO(2)$-action on two of the three copies (for fermions, up to an irrelevant parity factor).
In particular, the projections onto the corresponding invariant subspaces act as~$P$ on the respective pairs of copies.

Thus we have the following general situation:
Let $\pi \colon \SO(3) \to \U(\HH)$ be a (strongly) continuous representation on a (possibly infinite-dimensional) Hilbert space~$\HH$.
Let $\Pi_x$, $\Pi_y$, and $\Pi_z$ denote the projections on the subspaces invariant under rotations about each of the coordinate axes.
Then we wish to compute the spectral gap of the operator
\begin{align}\label{eq:def-T}
    T_\pi = \frac13 \left( \Pi_x + \Pi_y + \Pi_z \right).
\end{align}
For the natural action of $\SO(3)$ on $\HH = \mathrm{L}^2(\S^2)$, where $\S^2 \subseteq \RR^3$ is the two-dimensional unit sphere, the operator $T_\pi$ is precisely the transition operator of the Kac walk.
Its spectral gap has been studied in~\cite{carlen2000many,diaconis2000bounds,carlen2003determination,maslen2003eigenvalues,caputo2008spectral} more generally for the Kac walk on~$\S^{t-1}$: for $t=3$ specifically it is given by~$5/12$.
See, e.g., \cite[Thm.~3.1, $N=3$]{carlen2000many} for the proof under a simplifying assumption (permutation-symmetric functions) and \cite[Thm.~1.3, $N=3$ and $\rho=1/(2\pi)$]{carlen2003determination} for the general case.
Since the representation $\mathrm{L}^2(\S^2)$ contains \emph{all} irreducible representations of $\SO(3)$, by the Peter--Weyl theorem or concretely as the spherical harmonics, this lower bounds the spectral gap for an arbitrary representation.
Thus one obtains the following theorem.
To keep the paper self-contained, we give a simple derivation using basic representation theory for the convenience of the reader.

\begin{theorem}[{e.g., \cite[Thm.~1.3, $N=3$, $\rho=\frac1{2\pi}$]{carlen2003determination}}]\label{thm:kac-walk}
    The transition operator of the Kac walk on $\S^2$ has spectral gap $5/12$.
    As a consequence, for any representation $\pi$, we have $\spec(T_\pi) \subseteq [0,\frac7{12}] \cup \{1\}$.
\end{theorem}
\begin{proof}
    For a trivial representation, $T_\pi = I$ and hence has spectrum $\{1\}$.
    We will show that as $\pi$ ranges over the nontrivial irreducible representations of $\SO(3)$, the largest eigenvalue of~$T_\pi$ is~$7/12$.
    Recall that the irreducible representations of $\SO(3)$ correspond to the irreducible $\SU(2)$-representations with \emph{integer} spin~$j\in\NN$, i.e., the restriction of the $(2j)$-th tensor action on $\HH = \Sym^{2j}(\CC^2)$.
    We are interested in the case $j>0$.
    The three subgroups isomorphic to $\SO(2) \cong \U(1)$ corresponding to the coordinate axis rotations are generated by ($\mathrm{i}$ times) the Pauli matrices~$X,Y,Z$.
    The latter act on~$\HH$ with eigenvalues~$\{-2j,-2j+2,\dots,2j\}$; each eigenspace has multiplicity one.
    Thus, $\Pi_z = \ket{0_z}\!\!\bra{0_z}$ is the projection onto the unit vector $\ket{0_z} \propto \ket0^{\ot j}\ket1^{\ot j} + \emph{\text{permutations}}$, which spans the zero eigenspace for the action of~$Z$;
    $\Pi_x = \ket{0_x}\!\!\bra{0_x}$, where $\ket{0_x} \coloneqq R^{\ot 2j} \ket{0_z} \propto \ket+^{\ot j} \ket-^{\ot j} + \emph{\text{permutations}}$ and $R = \frac1{\sqrt2}\begin{psmallmatrix}1 & -1\\1 & 1\end{psmallmatrix}$, since $R$ conjugates $Z$ to $X$;
    and similarly $\Pi_y = \ket{0_y}\!\!\bra{0_y}$, where $\ket{0_y} \coloneqq S^{\ot 2j} \ket{0_x}$ and $S = \operatorname{diag}(e^{-\mathrm{i}\pi/4}, e^{\mathrm{i}\pi/4})$.
    Thus:
    \begin{align*}
        T_\pi = \frac13 \left( \ket{0_x}\!\!\bra{0_x} + \ket{0_y}\!\!\bra{0_y} + \ket{0_z}\!\!\bra{0_z} \right).
    \end{align*}
    The nonzero eigenvalues of $T_\pi$ can be computed in terms of the Gram matrix of the three unit vectors.
    By symmetry, their inner products are all equal, and a short calculation reveals that
    \begin{align*}
        \ip{0_x|0_y} = \ip{0_y|0_z} = \ip{0_x|0_z} = a_j \coloneqq 2^{-j}\sum_{k=0}^j(-1)^k\binom jk^2
    = \begin{cases}
        0 & \text{if $j$ is odd}, \\
        (-1)^{j/2} 2^{-j} \binom j {j/2} & \text{if $j$ is even}.
        \end{cases}
    \end{align*}
    Hence the maximal eigenvalue of $T_\pi$ equals the maximal eigenvalue of
    \begin{align*}
        \frac13 \begin{pmatrix}
        1 & a_j & a_j \\
        a_j & 1 & a_j \\
        a_j & a_j & 1
        \end{pmatrix}.
    \end{align*}
    If $j$ is odd, then this is simply $I/3$, so the maximal eigenvalue is $1/3 < 7/12$.
    Now suppose that $j$ is even.
    Then the above matrix has eigenvalues $\frac{1 + 2a_j}3$ and $\frac{1-a_j}3$ (with multiplicity two), and its maximal eigenvalue is
    $\lambda_j \coloneqq \max \braces{ \frac{1 + 2a_j}3, \frac{1 - a_j}3 }$.
    For $j=2$, we have $a_2 = -\frac12$ and hence $\lambda_2 = \frac12$.
    For $j=4$, $a_4 = \frac38$ and hence $\lambda_4 = \frac7{12}$, which attains the claimed upper bound.
    For $j\geq6$, we can estimate
    \begin{align*}
        \lambda_j \leq \frac{1 + 2\abs{a_j}}3 < \frac{1 + 2 \abs{a_4}}3 = \lambda_4
    \end{align*}
    because
    \begin{align*}
        \abs*{\frac {a_{j+2}} {a_j}}
    = \frac {\binom{j+2} {j/2+1}} {4 \binom j {j/2}}
    = \frac {j+1}{j+2} < 1
    \end{align*}
    and hence the $\abs{a_j}$ for even $j$ are monotonically decreasing.
    This concludes the proof.
\end{proof}

\section{Testing fermionic Gaussian states}\label{sec:fermions}
We recall the notion of fermions and fermionic Gaussian states and unitaries; see, e.g.,~\cite{bravyi2005lagrangian}.
For $m$ fermionic modes, the Hilbert space is the \emph{fermionic Fock space}
\begin{align*}
    \HH_m
    \coloneqq \mathcal F_\wedge(\CC^m)
    \coloneqq \bigoplus_{k=0}^m\bigwedge^k\CC^m.
\end{align*}
The fermionic creation and annihilation operators~$a_j^\dagger$ and~$a_j$, for $j\in\{1,\ldots,m\}$, act on $\HH_m$ and satisfy the \emph{canonical anticommutation relations} $\{a_j,a_k^\dagger\}=\delta_{jk}\id$ and $\{a_j,a_k\}=\{a_j^\dagger,a_k^\dagger\}=0$ for all $j,k$.
The number operators are $N_j \coloneqq a_j^\dagger a_j$, and $N \coloneqq \sum_jN_j$ is called the total particle number operator.
The \emph{Majorana operators}~$c_1,\dots,c_{2m}$ are defined by $c_{2j-1} \coloneqq a_j^\dagger+a_j$ and $c_{2j} \coloneqq \mathrm{i}(a_j^\dagger-a_j)$.
They are Hermitian unitaries that pairwise anticommute, i.e.,
\[ c_j=c_j^\dagger \quad\text{and}\quad \{c_j,c_k\}=2\delta_{jk}\id. \]
Finally, the parity operator is given by $\Pi \coloneqq (-1)^N = (-\mathrm{i})^m\prod_{j=1}^{2m} c_j$.
It is also a Hermitian unitary, and it anticommutes with every Majorana operator.
We write $\HH_m=\HH_m^+\op\HH_m^-$ for the decomposition of the Hilbert space into the $\pm1$ eigenspaces of~$\Pi$, called the even and odd parity subspaces, respectively.

For every $R\in\SO(2m)$, there is a unitary $U_R$ such that
\begin{align}\label{eq:action fermionic gaussian unitary}
    U_Rc_jU_R^\dagger
    =\sum_{k=1}^{2m}R_{kj}c_k.
\end{align}
This gives a projective representation of $\SO(2m)$ on~$\HH_m$, called the \emph{spin representation}; it lifts to a representation of the double cover~$\Spin(2m)$.
The unitaries~$U_R$ are called \emph{Gaussian unitaries}.
They are generated by Hamiltonians quadratic in the Majorana operators, $H=\mathrm{i}\sum_{j,k=1}^{2m}h_{jk}c_jc_k$, where~$h_{jk} = -h_{kj} \in \RR$, i.e., $h \in \mathfrak{so}(2m)$.
For the subgroup $\SU(m) \subseteq \SO(2m)$, the corresponding Gaussian unitaries are given by the second quantization action~$\Gamma_{\wedge}(u)$ of unitaries~$u\in\SU(m)$, which acts as the restriction of $u^{\ot k}$ on each fixed particle number subspace $\bigwedge^k \CC^m$.

A pure state is called a \emph{fermionic Gaussian state} if it is the unique ground state of a quadratic Hamiltonian.
Within either parity sector, the pure Gaussian states form a single orbit of the Gaussian unitaries, which we denote by~$\GG_m^+$ and~$\GG_m^-$, and we let $\GG_m\coloneqq\GG_m^+\cup\GG_m^-$ denote the set of all pure Gaussian states.

Concretely, one can identify $\HH_m \cong (\CC^2)^{\ot m}$ using a Jordan--Wigner transformation, which realizes the Majorana operators as certain Pauli operators:
$c_{2k-1} = Z_1 \cdots Z_{k-1} X_k$ and $c_{2k} = Z_1 \cdots Z_{k-1} Y_k$.
Then, $\mathrm{i} c_{2k-1} c_{2k} = i X_k Y_k = -Z_k = 2 N_k - I$.
In particular, the standard basis vectors~$\ket x$ for $x \in \{0,1\}^m$ form a joint eigenbasis of the number operators~$N_j$: we have $N_j \ket x = x_j \ket x$.
The even (odd) Gaussian states are the orbit of any standard basis state with even (odd) Hamming weight under the Gaussian unitaries.


\subsection{A condition for fermionic Gaussianity}
We consider a two-copy test based on a well-known condition for fermionic Gaussianity.
On~$\HH_m^{\ot 2}$, define
\begin{align*}
    \Lambda
    \coloneqq\sum_{j=1}^{2m}c_j\ot c_j = 2\sum_{j=1}^m (a_j^\dagger \ot a_j + a_j \ot a_j^\dagger).
\end{align*}
Let $P$ be the orthogonal projection onto the kernel of $\Lambda$.
This defines the accepting outcome of a binary measurement, with acceptance probability
\[ p(\psi) = \langle P \rangle_{\psi^{\ot 2}} = \angles{\psi^{\ot 2} | P | \psi^{\ot 2}} \]
as in \cref{eq:p-accept-2-copies}.
We note a few basic properties.
First of all, $P$ commutes with the swap operator on the copies (in fact, it is supported on the symmetric subspace, see below).
Simultaneous Gaussian conjugation leaves~$\Lambda$ invariant, as is readily verified using \cref{eq:action fermionic gaussian unitary}:
\begin{align}\label{eq:Lambda-Gaussian-invariance}
    U_R^{\ot 2}  \Lambda  U_R^{\dagger, \ot 2} = \Lambda.
\end{align}
Thus the acceptance probability $p(\psi)$ is invariant under Gaussian unitaries.
The operator $\Lambda$ is well-known to characterize Gaussian states~\cite{bravyi2005lagrangian,melo2013power}.
We include the short proof which does not assume a priori that the state has fixed parity.

\begin{lemma}[{cf.~\cite[Cor.~1]{melo2013power} for fixed parity}]\label{lem:exact fermions}
Let $\psi\in\P(\HH_m)$.
The following are equivalent:
\begin{enumerate}[(i),itemsep=0pt]
    \item\label{it:f-gaussian} $\psi$ is a pure fermionic Gaussian state (i.e., $\psi \in \GG_m$),
    \item\label{it:f-lambda} $\Lambda\ket\psi^{\ot 2}=0$,
    \item\label{it:f-p} $p(\psi)=1$.
\end{enumerate}
\end{lemma}
\begin{proof}
The implication \ref{it:f-gaussian}~$\Rightarrow$~\ref{it:f-lambda} follows from \cref{eq:Lambda-Gaussian-invariance} and $\Lambda \ket{x}^{\ot 2} = 0$ for any standard basis state.
The implication \ref{it:f-lambda}~$\Rightarrow$~\ref{it:f-p} holds by definition, since~$P$ is the projection onto the kernel of~$\Lambda$.
Finally, for \ref{it:f-p}~$\Rightarrow$~\ref{it:f-gaussian}, we have
\begin{align}\label{eq:norm-Gamma}
    0 = \norm{\Lambda \ket{\psi}^{\ot 2}}^2 = 2m - \norm{\Gamma_\psi}_2^2
\end{align}
where the first equality holds by assumption and the second by a simple calculation; $\norm{\cdot}_2$ denotes the Hilbert--Schmidt norm and $(\Gamma_\psi)_{jk} \coloneqq \frac{\mathrm{i}}{2}\langle [c_j, c_k] \rangle_\psi$ is the so-called \emph{covariance (or correlation) matrix}.
The covariance matrix is real and antisymmetric, hence it can be block diagonalized by a rotation (which does not change the norm).
Hence, by acting with a Gaussian unitary, we may assume without loss of generality that $\Gamma_\psi$ is block diagonal, with blocks
    \begin{align*}
        \Gamma_\psi = \bigoplus_{k=1}^m \begin{pmatrix}
            0 & \nu_k \\ -\nu_k & 0
        \end{pmatrix} \, \qquad \nu_k \in \RR.
    \end{align*}
    Since $\nu_k = i\langle c_{2k-1} c_{2k} \rangle_\psi = -\langle Z_k \rangle_\psi$, we know that in fact $\abs{\nu_k} \leq 1$.
    Now \cref{eq:norm-Gamma} implies that $\nu_k = \pm 1$.
    Thus the state is a joint eigenstate of $Z_1, \dots, Z_m$, i.e., a standard basis state, hence a Gaussian state.
\end{proof}

\noindent
In fact, it is known that $\ker \Lambda = \Span \{ \ket{\psi}^{\ot 2} : \ket{\psi} \in \GG_m \}$; see \cite[Lem.~A.1]{melo2013power} with twirling over the full orthogonal group.
In particular, $P$ is supported on the symmetric subspace, as asserted earlier (but we will not require this fact).

\smallskip

The operator $\Lambda$ has an interpretation in terms of Gaussian unitaries.
To see this, let us define Majorana operators on $2m$ modes, corresponding to an isomorphism $\HH_m^{\ot 2} \cong \HH_{2m} \cong \mathcal{F}_{\wedge}(\CC^2 \ot \CC^m)$:
\begin{align}\label{eq:two copy majoranas}
    c_{j,1} = c_j\ot\id, \qquad
    c_{j,2} = \Pi\ot c_j,
    \qquad (j=1,\dots,2m)
\end{align}
which are easily seen to satisfy the canonical anticommutation relations $\{c_{j,a},c_{k,b}\} = 2\delta_{jk}\delta_{ab}\id$.
Then:
\begin{align}\label{eq:two copy majorana rotation}
    L
\coloneqq \mathrm{i} \sum_{j=1}^{2m} c_{j,1} c_{j,2}
 = -\mathrm{i} (\Pi \ot I) \Lambda
\end{align}
is a generator of rotations between the two sets of $m$~modes; these are the Gaussian unitaries corresponding to the subgroup $\SO(2) \subseteq \SU(\CC^2 \ot \CC^m) = \SU(2m) \subseteq \SO(4m)$, given simply by the second quantized action~$\Gamma_{\wedge}(R \ot I_m)$ for $R \in \SO(2)$, as explained above.
In particular, $\ker L = \ker \Lambda$ because $\Pi$ is invertible, so we see that~$P$ identifies with the projection onto the rotation-invariant subspace of~$\HH_{2m}$.

This is a special case of a more general symmetry: if we have~$t$ copies, we can likewise identify $\HH_m^{\ot t} \cong \HH_{tm} \cong \mathcal{F}_{\wedge}(\CC^t \ot \CC^m)$, and there is an action of $\Ortho(t)$ given by the second quantized action~$\Gamma_{\wedge}(R \ot I_m)$ for $R \in \Ortho(t)$.
As discussed in the introduction, this action spans the commutant of the tensor power action of the fermionic Gaussian unitaries, by \emph{Howe duality}~\cite{howe1989remarks}.

\subsection{Testing fermionic Gaussianity with a constant number of copies}
The two-outcome measurement $\{P,\id-P\}$ accepts two copies of an unknown state with probability one if and only if the state is Gaussian.
We now show that this characterization is \emph{robust}: states that are accepted with high probability are close to Gaussian states.

\begin{theorem}\label{thm:main-fermions}
Let $\psi\in\P(\HH_m)$. If $p(\psi) \geq 1 - \tau$, then $F_{\GG_m}(\psi)^2 \geq 1 - 8\tau$.
\end{theorem}
\begin{proof}
We will use the method in \cref{sec:gradient-flow,sec:kac} and bound the spectral gap of the operator~$S$ in \cref{eq:def-S} by identifying it with the operator~$T_\pi$ in \cref{eq:def-T} for a suitable representation of $\SO(3)$.
Recall that
\begin{align*}
    S = \frac13\left(P_{12} + P_{13} + P_{23} \right),
\end{align*}
where $P_{ab}$ denotes the operator that acts as $P$ on subsystems~$a$ and~$b$.
Analogously to \cref{eq:two copy majoranas}, we define Majorana operators on $3m$ modes, corresponding to an isomorphism $\HH_m^{\ot 3} \cong \HH_{3m} \cong \mathcal F_{\wedge}(\CC^3 \ot \CC^m)$:
\begin{align*}
    c_{j,1}=c_j\ot\id\ot\id, \qquad
    c_{j,2}=\Pi\ot c_j\ot\id, \qquad
    c_{j,3}=\Pi\ot\Pi\ot c_j,
    \qquad (j=1,\dots,2m)
\end{align*}
which satisfy the canonical anticommutation relations $\{c_{j,a},c_{k,b}\} =2\delta_{jk}\delta_{ab}\id$.
On $\HH_{3m}$, we have a representation $\pi$ of $\SO(3) \subseteq \SU(\CC^3 \ot \CC^m) = \SU(3m) \subseteq \SO(6m)$, given by the second quantized action of~$\Gamma_{\wedge}(R \ot I_m)$ for $R \in \SO(3)$.
Corresponding to the rotations about the three coordinate axes, similarly as in \cref{eq:two copy majorana rotation}, we have generators
\begin{align*}
    L_{ab} = \mathrm{i} \sum_{j=1}^{2m} c_{j,a} c_{j,b}.
\end{align*}
Then, $L_{12} = - \mathrm{i} (\Pi \ot I \ot I) \Lambda_{12}$, $L_{23} = - \mathrm{i} (I \ot \Pi \ot I) \Lambda_{23}$, and $L_{13} = - \mathrm{i} (\Pi \ot \Pi \ot I) \Lambda_{13}$, where $\Lambda_{ab}$ denotes the operator that acts as $\Lambda$ on subsystems~$a$ and~$b$.
Thus, the projections~$P_{ab}$ identify with the projections~$\Pi_x,\Pi_y,\Pi_z$ onto the subspaces invariant under the three coordinate axis rotations, as in \cref{eq:def-T}.
It follows that $S = T_\pi$ under our identification.
By \cref{thm:kac-walk}, we have $\spec(S) \subseteq [0,\frac{7}{12}] \cup \{1\}$, and now \cref{thm:bound-fidelity-gradient-flow} concludes the proof because $C_{7/12} = 1/(1-\frac78) = 8$.
\end{proof}

As a consequence, we find that the number of copies needed to test whether a state is Gaussian or at least constantly far from Gaussian is \emph{independent} of the number of modes~$m$.
Indeed, one simply repeats the two-outcome measurement a suitable number of times and accepts if all measurements accept.

\begin{corollary}\label{cor:fermion-testing}
For any $\eps,\delta\in(0,1)$, there exists a test that uses $\bigO(\frac1{\eps^2} \log \frac1\delta)$ copies of a pure state~$\psi$ to achieve the following:
if $\psi$ is a fermionic Gaussian state, the test always accepts, while if $F_{\GG_m}(\psi)^2 \leq 1-\eps^2$, it rejects with probability at least~$1-\delta$.
The test acts on two copies of the state at a time.
\end{corollary}
\begin{proof}
Repeat the two-outcome measurement $n = \ceil{\frac8{\eps^2}\log\frac1\delta}$ times, on fresh pairs of copies, and accept if every measurement accepts.
In total, this uses $2n$ copies of the unknown state~$\psi$.
If $\psi$ is fermionic Gaussian, then this test always accepts by \cref{lem:exact fermions}.
If $F^2_{\GG_m}(\psi) \leq 1 - \eps^2$, then $p(\psi) \leq 1 - \frac{\eps^2}8$ by \cref{thm:main-fermions}, and hence the test accepts with probability at most
$p(\psi)^n \leq \parens*{ 1 - \tfrac{\eps^2}8 }^n \leq e^{-n\tfrac{\eps^2}8} \leq \delta$.
\end{proof}

\section{Testing bosonic Gaussian states}\label{sec:bosons}
The Hilbert space of $m$ bosonic modes is the \emph{bosonic Fock space}
\begin{align*}
    \HH_m = L^2(\RR^m) \cong \mathcal F_\vee \CC^m = \bigoplus_{k=0}^{\infty} \Sym^k(\CC^m).
\end{align*}
The position and momentum operators $X_j$ and~$P_j$ ($j=1,\ldots,m$) are defined by~$(X_j f)(x) = x_j f(x)$ and~$(P_j f)(x) = -\mathrm{i} \frac{\partial f (x)}{\partial x_j}$, respectively.
They satisfy $[X_j,P_k] = \mathrm{i}\,\delta_{jk} \id$.
The bosonic creation and annihilation operators are defined by $a_j^\dagger \coloneqq (X_j - \mathrm{i} P_j)/\sqrt2$ and $a_j \coloneqq (X_j + \mathrm{i} P_j)/\sqrt2$, and they satisfy the \emph{canonical commutation relations}:
\[ [a_j, a_k^\dagger] = \delta_{jk}I \quad\text{and}\quad [a_j, a_k] = 0 = [a_j^\dagger,a_k^\dagger]. \]
The number operator in the $j$-th mode is defined as $N_j = a_j^\dagger a_j$.
There is a standard \emph{occupation number basis} of $\HH_m$, consisting of joint eigenvectors of the number operators $N_j$, given by
\begin{align*}
    \ket{n} \coloneqq \frac{(a_1^\dagger)^{n_1} \dots (a_m^\dagger)^{n_m}}{\sqrt{(n_1!) \cdots (n_m!)}} \ket{0} \qquad (n \in \NN^m),
\end{align*}
where $\ket{0}$ is the vacuum state.
Write $r=(x,p)\in \RR^{2m}$ and $R = (X_1,\ldots,X_m,P_1,\ldots,P_m)$.
For any state~$\rho \in \S(\HH_m)$, we call $\tr[\rho R]$ (when it exists) its \emph{mean}.
Let $\Omega \coloneqq \begin{psmallmatrix}0 & I\\-I & 0\end{psmallmatrix}$ denote the matrix defining the standard symplectic form on~$\RR^{2m}$, so $[r^TR, s^TR] = \mathrm{i} \, r^T\Omega s \id$.
The \emph{displacement (or Weyl) operators} are the unitary operators defined by $D_r \coloneqq \exp( -\mathrm{i}\, r^T\Omega R)$, or equivalently by their action on wavefunctions, $(D_{(x,p)} f)(z) = \exp(\mathrm{i}\,p^T(z-\frac{x}{2})) f(z-x)$ for $x,p,z\in\RR^m$.
The \emph{characteristic function} of a state~$\rho \in \S(\HH_m)$ is defined as follows:
\begin{align*}
    \chi_\rho \colon \RR^{2m}\to \CC, \qquad \chi_\rho(r)=\tr[\rho D_r].
\end{align*}
It satisfies $\chi_\rho(0)=1$, is continuous at $0$ and bounded by $\abs{\chi_\rho(r)} \leq 1$, satisfies $\chi_\rho(-r)=\overline{\chi_\rho(r)}$, and uniquely determines $\rho$.

For every symplectic matrix $S \in \Sp(2m)$, there is a unitary~$U_S$ such that
\begin{align*}
    U_S^\dagger R U_S = S R, \qquad U_S^\dagger D_r U_S = D_{S^{-1}r}.
\end{align*}
This gives a projective representation of the symplectic group $\Sp(2m)$ on $\HH_m$;
it lifts to a representation of the metaplectic group.
The unitaries $U_S$ are called \emph{(symplectic) Gaussian unitaries}.
They are generated by Hamiltonians homogeneously quadratic in the position and momentum operators, $H = \sum_{j,k=1}^{2m} h_{jk} R_j R_k$, where $\overline{h_{jk}} = h_{kj} \in \CC$.
For the subgroup $\U(m) \subseteq \Sp(2m)$, the corresponding Gaussian unitaries are given by the second quantization action~$\Gamma_{\vee}(u)$ of unitaries~$u\in\U(m)$, which acts as the restriction of $u^{\ot k}$ on each fixed particle number subspace $\Sym^k \CC^m$.

A pure state $\psi$ is called a \emph{bosonic Gaussian state} if it is the unique ground state of a (not necessarily homogeneous) quadratic Hamiltonian.
Equivalently, the wave function takes the form $\psi(x) = \exp(P(x))$ where $P$ is a complex polynomial of degree two, or $\ket{\psi} \propto D_r U_S \ket{0}$ where $\ket{0}$ is the vacuum state, for some $r \in \RR^{2m}$ and $S \in \Sp(2m)$.
A final equivalent characterization is that a pure state is Gaussian if and only if its characteristic function is Gaussian.
We denote the set of Gaussian states by~$\GG_m$.
We also consider the set of pure \emph{bosonic Gaussian states with zero mean}, denoted $\GG_m^0$; these are precisely the states of the form~$\ket\psi \propto U_S \ket{0}$ for $S \in \Sp(2m)$.
Finally, a \emph{coherent state} is a state of the form $D_r \ket{0}$ for some $r \in \RR^{2m}$; we denote by $\mathcal C_m$ the set of coherent states.

\subsection{Conditions for bosonic Gaussianity}\label{sub:conditions_for_bosonic_gaussianity}
We consider a two-copy test based on a well-known condition for bosonic Gaussian states with zero mean, closely analogous to the fermionic case.
On $\HH_m^{\ot 2}$, define
\begin{equation*}
    \Lambda \coloneqq \mathrm{i} \sum_{j=1}^m \left( a_j \ot a_j^\dagger - a_j^\dagger \ot a_j \right).
\end{equation*}
To be precise, $\Lambda$ is an unbounded operator, and we need to specify its domain~$\dom(\Lambda)$.
We take this to be the closure of the space of finite linear combinations of Fock basis vectors with respect to the graph norm of~$\Lambda$ (i.e., $\norm{f}_{\Lambda} \coloneqq \norm{f} + \norm{\Lambda f}$).
The operator $\Lambda$ is the generator of rotations between the two copies: $U_\theta \coloneqq \exp(-i \Lambda \theta)$ acts on a wavefunction $\Psi \in \HH_m^{\ot 2} \cong \HH_{2m} \cong L^2(\RR^m \op \RR^m)$ as
\begin{align}\label{eq:U-theta}
    (U_\theta \Psi)(x,y) = \Psi(\cos(\theta)x + \sin(\theta)y, -\sin(\theta)x + \cos(\theta) y).
\end{align}
In other words, it is the generator of rotations between the two sets of $m$~bosonic modes;
these are the Gaussian unitaries corresponding to the subgroup $\SO(2) \subseteq \U(\CC^2 \ot \CC^m) = \U(2m) \subseteq \Sp(4m)$, given by the second quantized action~$\Gamma_\vee(R \ot I_m)$ for $R \in \SO(2)$, similarly to the fermionic case.
Just like in the fermionic case, this is a special case of a more general symmetry: if we have~$t$ copies, we can likewise identify $\HH_m^{\ot t} \cong \HH_{tm} \cong L^2(\RR^t \ot \RR^m)$, and there is an action of $\Ortho(t)$ by acting as $R \ot I_m$ for $R \in \Ortho(t)$ on the coordinates of the wave function: $(\pi(R)\Psi)(z) \coloneqq \Psi((R^{-1} \ot I_m)z)$, or by second quantization as above.
As discussed, this action generates the commutant of the tensor power action of the symplectic Gaussian unitaries~$\Sp(2m)$, by \emph{Kashiwara--Vergne duality}~\cite{kashiwara1978segal}.

Let $P_0$ be the orthogonal projection onto the kernel of $\Lambda$.
This defines the accepting outcome of a binary measurement with acceptance probability
\[ p_0(\psi) \coloneqq \ip{P_0}_{\psi^{\ot 2}} = \angles{\psi^{\ot 2} | P_0 | \psi^{\ot 2}}. \]
In particular, if $\ket\psi^{\ot2}$ is in the kernel of $\Lambda$, it is in $\dom(\Lambda)$.
Note that the test is defined for any input state, without any moment assumptions.
We again note some basic properties.
First, $P_0$ commutes with the swap operator between the two copies.
It is not hard to see $\Lambda$ is invariant under conjugation by Gaussian unitaries
\begin{align}\label{eq:lambda-commutes-gaussian}
    U_S^{\ot 2}  \Lambda U_S^{\dagger, \ot 2} = \Lambda.
\end{align}
Thus the same is true for the acceptance probability~$p_0(\psi)$.
The operator $\Lambda$ characterizes (mean-zero) Gaussian states in a similar fashion to the fermionic case~\cite{girardi2025gaussian}.
We include the short proof for convenience.

\begin{proposition}\label{prp:mean-zero-gaussian}
    Let $\ket{\psi} \in \HH_m$ be a pure state. The following are equivalent:
    \begin{enumerate}[(i),itemsep=0pt]
        \item\label{it:state-gaussian} $\psi$ is a pure Gaussian state with zero mean (i.e., $\psi \in \GG_m^0$),
        \item\label{it:lambda-annihilates} $\Lambda \ket{\psi}^{\ot 2} = 0$,
        \item\label{it:p-acc-one} $p_0(\psi) = 1$.
    \end{enumerate}
\end{proposition}

\begin{proof}
    We have \ref{it:lambda-annihilates}~$\Leftrightarrow$~\ref{it:p-acc-one} by definition.
    For \ref{it:state-gaussian}~$\Rightarrow$~\ref{it:lambda-annihilates}, by \cref{eq:lambda-commutes-gaussian} it suffices to show $\Lambda \ket{\psi}^{\ot 2} = 0$ for the vacuum state~$\ket{\psi} = \ket{0}$, which follows immediately from $a_j \ket{0} = 0$.
    Finally, we show that \ref{it:lambda-annihilates}~$\Rightarrow$~\ref{it:state-gaussian}.
    If~$p_0(\psi)=1$, then, in particular, $U_\theta \ket{\psi}^{\ot 2} = \ket{\psi}^{\ot 2}$ for $\theta = \pi/4$.
    At the level of the characteristic function~$\chi$ of $\psi$, this translates to
    \begin{align}\label{eq:chi-rotation}
        \chi(r)\chi(s) = \chi\mleft(\frac{r+s}{\sqrt2}\mright)\chi\mleft(\frac{r-s}{\sqrt2}\mright) \qquad (r, s \in \RR^{2m}).
    \end{align}
    If we take $r=0$, we get $\chi(s) = \abs{\chi(s/\sqrt2)}^2$, and hence iterating $\chi(r) = \abs{\chi(2^{-k/2} r)}^{2^k}$.
    In particular, by continuity and $\chi(0) = 1$, $\chi$ is strictly positive.
    If we let $f(r) \coloneqq \log \chi(r)$, then \cref{eq:chi-rotation} gives $f(r) + f(s) = f((r+s)/\sqrt2) + f((r-s)/\sqrt2)$.
    In particular, $f(r) = 2 f(r/\sqrt2)$.
    If we let $B(r,s) \coloneqq \frac14(f(r+s) - f(r-s))$, then $B$ is additive in each variable; 
    by continuity it is in fact bilinear.
    We conclude that $f(r) = B(r,r)$ is quadratic, hence~$\chi$ is a centered Gaussian, and we conclude that $\ket{\psi}$ is a Gaussian state with zero mean.
\end{proof}

A useful fact (see also \cite[Lem.~30]{girardi2025gaussian}) is that if we apply the test to mixed states, high acceptance probability implies that the state is close to a pure state.

\begin{lemma}\label{lem:purity}
    Let $\rho \in \S(\HH_m)$.
    If $\tr[P_0 \rho^{\ot 2}] \geq 1 - \tau$, then $\norm{\rho}_{\infty} \geq 1 - 2\tau$.
\end{lemma}
\begin{proof}
    We have
    \begin{align*}
        \abs{\tr[U_{\pi/2} \rho^{\ot 2}]}
      \geq \tr[P_0\rho^{\ot 2}] - \tr[(1-P_0)\rho^{\ot 2}] \geq 1 - 2\tau,
    \end{align*}
    since the image of $P_0$ is in the $+1$-eigenspace of $U_{\pi/2}$.

    On the other hand, the rotation $U_{\pi/2}$ is equivalent to a swap of the two copies and a reflection on the second copy (defined by $J f(x) = f(-x)$ for $f \in \HH_m = L^2(\RR^m)$), so
    \begin{align*}
        \tr[U_{\pi/2} \rho^{\ot 2}] = \tr[J\rho^2].
    \end{align*}
    By Cauchy-Schwarz, using the unitarity of~$J$,
    \begin{align*}
        \abs{\tr[J\rho^2]} \leq \tr[\rho^2] \leq \norm{\rho}_{\infty}.
    \end{align*}
    The lemma follows from combining these three equations.
\end{proof}

Next, we consider a condition for Gaussianity \emph{with arbitrary mean}.%
\footnote{Another approach would be to consider invariance under stochastic orthogonal transformations, which generate the commutant of the tensor-power action of arbitrary Gaussian unitaries (including displacements).
This leads to a natural three-copy test, cf.~\cite{girardi2025gaussian}.
Here we proceed differently, in a way that naturally relates to a test for coherent states, which might be of independent interest.}
Let~$U_{\pi/4}$ be as in \cref{eq:U-theta}.
Then:
\begin{align}\label{eq:displacement-rotation}
    U_{\pi/4} (D_r \ot D_s) U_{\pi/4}^\dagger = D_{(r - s)/\sqrt2} \ot D_{(r+s)/\sqrt2}.
\end{align}
We now consider a four-copy test, where the accepting projection operator is
\begin{align*}
    P \coloneqq U_{\pi/4}^{\ot 2} P_{0,24} U_{\pi/4}^{\dagger, \ot 2}
\end{align*}
where $P_{0,24}$ denotes the two-copy operator~$P_0$ (the projection onto the kernel of~$\Lambda$) acting on the second and fourth copies of $\HH_m$.
Denote by
\begin{align}\label{eq:four copy prob}
  p(\psi)
\coloneqq \ip{P}_{\psi^{\ot 4}}
= \angles{\psi^{\ot 4} | P | \psi^{\ot 4}}
= \tr[P_0 \rho^{\ot 2}]
\end{align}
the acceptance probability of the test, where $\rho = \tr_1[U_{\pi/4}^\dagger \psi^{\ot 2} U_{\pi/4}]$.
We also consider a closely related two-copy test for \emph{coherent states}, with accepting projection operator given by
\begin{align*}
    Q \coloneqq U_{\pi/4} (I \ot \proj{0}) U_{\pi/4}^\dagger.
\end{align*}
Note that $Q$ commutes with the swap operator between the two copies.
Denote the acceptance probability of the latter test by
\begin{align}\label{eq:coherent prob}
  q(\psi)
\coloneqq \ip{Q}_{\psi^{\ot 2}}
= \ip{0 | \rho | 0},
\end{align}
where again $\rho = \tr_1[U_{\pi/4}^\dagger \psi^{\ot 2} U_{\pi/4}]$.

\begin{proposition}
    Let $\ket{\psi} \in \HH_m$ be a pure state. The following are equivalent:
    \begin{enumerate}[(i),itemsep=0pt]
        \item\label{it:state-gaussian-nonzero-mean} $\ket{\psi}$ is a Gaussian state,
        \item\label{it:p-acc-one-nonzero-mean} $p(\psi) = 1$.
    \end{enumerate}
    Furthermore, the following are equivalent:
    \begin{enumerate}[(i),itemsep=0pt]
        \item\label{it:state-coherent} $\ket{\psi}$ is a coherent state,
        \item\label{it:p-acc-one-coherent} $q(\psi) = 1$.
    \end{enumerate}
\end{proposition}

\begin{proof}
    If $\ket{\psi_r}$ is a Gaussian state with mean $r \in \RR^{2m}$, then $\ket{\psi_0} \coloneqq D_{-r} \ket{\psi_r}$ is Gaussian with zero mean.
    Consequently,
    \begin{align*}
        U_{\pi/4}^\dagger \ket{\psi_r}^{\ot 2} = U_{\pi/4}^\dagger D_r^{\ot 2} \ket{\psi_0}^{\ot 2} = (D_{\sqrt2 r} \ot I) U_{\pi/4}^\dagger \ket{\psi_0}^{\ot 2} = (D_{\sqrt2 r} \ket{\psi_0}) \ot \ket{\psi_0}
        \qquad (r \in \RR^{2m})
    \end{align*}
    using \cref{eq:displacement-rotation,prp:mean-zero-gaussian}.
    Thus, after applying $(U^\dagger_{\pi/4})^{\ot 2}$ to $\ket{\psi_r}^{\ot 4}$, we have two copies of $\ket{\psi_0}$, one on the second and one on the fourth copy of the Hilbert space, and $P_{0,24}$ will accept with~probability~one.
    Similarly, if $\ket{\psi_r} = D_r \ket{0}$ is a coherent state, after applying $U_{\pi/4}^\dagger$ to two copies of $\ket{\psi_r}$, we are left with the vacuum state on the second copy, and thus $Q$ will accept~$\ket{\psi_r}^{\ot 2}$ with probability one.
    This shows (i)~$\Rightarrow$~(ii) for both families of states.

    For the converse, first assume that $p(\psi)=1$.
    Then $\tr[P_0\rho^{\ot 2}]=1$, where $\rho$ denotes the reduced state on the second register after applying $U_{\pi/4}^\dagger$ to $\ket{\psi}^{\ot 2}$.
    By \cref{lem:purity}, $\rho$ is pure, and then \cref{prp:mean-zero-gaussian} shows that it is Gaussian with zero mean.
    By using that $U_{\pi/4}$ commutes with $U_S^{\ot 2}$ for any symplectic Gaussian unitary~$U_S$, we may assume without loss of generality that this Gaussian state is the vacuum state.
    This reduces to the coherent state test:
    it suffices to show that if $q(\psi) = 1$, then $\ket{\psi}$ is a coherent state.
    In other words, if $U_{\pi/4}^\dagger \ket{\psi}^{\ot 2} = \ket{\phi} \ket{0}$ for some arbitrary state $\ket{\phi}$, then we wish to show that $\ket{\psi}$ is a coherent state.
    To that end, consider the characteristic functions $\chi_\psi$ and $\chi_\phi$, and let $\eta$ denote the characteristic function of the vacuum.
    Then, $U_{\pi/4}^\dagger \ket{\psi}^{\ot 2} = \ket{\phi} \ket{0}$ is equivalent to
        \begin{align*}
        \chi_\phi(r)\eta(s) = \chi_\psi\mleft(\frac{r+s}{\sqrt2}\mright)\chi_\psi\mleft(\frac{r-s}{\sqrt2}\mright) \qquad (s,r \in \RR^{2m}).
    \end{align*}
    Since $\eta(r)\eta(s) = \eta\mleft(\frac{r+s}{\sqrt2}\mright)\eta\mleft(\frac{r-s}{\sqrt2}\mright)$, if we let $\tilde \chi_\psi(r) \coloneqq \chi_\psi(r)/\eta(r)$, $\tilde \chi_\phi(r) \coloneqq \chi_\phi(r)/\eta(r)$, we get
    \begin{align}\label{eq:tilde-chi-rotation}
        \tilde\chi_\phi(r) = \tilde\chi_\psi\mleft(\frac{r+s}{\sqrt2}\mright)\tilde \chi_\psi\mleft(\frac{r-s}{\sqrt2}\mright) \qquad (s,r \in \RR^{2m}).
    \end{align}
    Setting $s=r$, we have $\tilde\chi_\phi(r) = \tilde \chi_\psi(\sqrt2 \, r)$, and we can rewrite \cref{eq:tilde-chi-rotation} as $\tilde \chi_\psi(r + s) = \tilde \chi_\psi(r) \tilde \chi_\psi(s)$.
    This implies $\abs{\tilde \chi_\psi(r)}^2 = \tilde \chi_\psi(r)\tilde \chi_\psi(-r) = \tilde \chi_\psi(0) = 1$.
    Thus, $\tilde \chi_\psi$ is a continuous homomorphism $\RR^{2m} \to U(1)$, and must be of the form $\tilde \chi_\psi(r) = e^{-if(r)}$ for a linear function~$f$.
    Hence, $\chi_\psi(r) = \eta(r) e^{-if(r)}$, which is the characteristic function of a coherent state.
\end{proof}

\subsection{Testing bosonic Gaussianity with a constant number of copies}
We will now show that the above tests are robust.
We start with the test for zero-mean Gaussian states.

\begin{theorem}\label{thm:main boson mean zero}
    Let $\psi \in \P(\HH_m)$.
    If $p_0(\psi) \geq 1 - \tau$, there is $\phi\in\GG_m^0$ such that $F^2_{\GG_m^0}(\psi) \geq \abs{\braket\phi\psi}^2 \geq 1-8\tau$.
\end{theorem}

\noindent
This is an easy consequence of \cref{thm:bound-fidelity-gradient-flow,thm:kac-walk}.
The argument is nearly identical to the fermionic case, except that we do not need to worry about the difference between ordinary and fermionic (graded) tensor products.

\begin{proof}[Proof of \cref{thm:main boson mean zero}]
    As explained in \cref{sub:conditions_for_bosonic_gaussianity}, there is an action~$\pi$ of $\SO(3)$ on $\HH_m^{\ot 3} \cong L^2(\RR^3 \ot \RR^m)$.
    In particular, we have the rotations $U_{\theta,ab}$ acting on copies $a$ and $b$ for $1 \leq a < b \leq 3$,
    with generators~$\Lambda_{ab}$ and projection~$P_{0,ab}$ onto their kernels.
    We let
    \begin{align*}
        S_0 \coloneqq \frac{1}{3}\left(P_{0,12} + P_{0,13} + P_{0,23}\right).
    \end{align*}
    By construction, $S_0 = T_\pi$, so \cref{thm:kac-walk} implies that $\spec(S_0) \subseteq [0,\frac{7}{12}] \cup \{1\}$.
    The result follows from \cref{thm:bound-fidelity-gradient-flow}.
\end{proof}

Next, we want to prove robustness of the test defined by $Q$ for coherent states.

\begin{theorem}\label{thm:main boson coherent}
    Let $\psi \in \P(\HH_m)$ be a pure state. If $q(\psi) \geq 1 - \tau$, then $F^2_{\mathcal C_m}(\psi) \geq 1 - 4\tau$.
\end{theorem}
\begin{proof}
This follows at once from \cref{thm:bound-fidelity-gradient-flow} and the following \cref{lem:S-spectrum-coherent-states}.
\end{proof}

\begin{lemma}\label{lem:S-spectrum-coherent-states}
For any number of modes~$m$, let $S \coloneqq \frac13 ( Q_{12} + Q_{13} + Q_{23} )$.
Then, $\spec(S)\subseteq[0,1/2]\cup\{1\}$.
\end{lemma}
\begin{proof}
    We identify
    \begin{align*}
        \HH_m^{\ot 3} &=
        \mathcal F_\vee(\CC^m)^{\ot 3} \cong \mathcal F_\vee(\CC^m \ot \CC^3).
    \end{align*}
    Then $Q_{12}$ equals the projection onto the states where the modes in $\CC^m \ot \CC^3$ of the form $x \ot (1,-1,0)^T$ are unoccupied, and similarly for $Q_{13}$ and $Q_{23}$.
    Thus, $Q_{ab}$ is the projection onto $\mathcal F_\vee(\CC^m \ot (V \op \CC w_{ab}))$, where~$V \coloneqq \CC v$, with $v \coloneqq (1,1,1)^T$ the uniform vector, and we have defined the following three vectors in~$W \coloneq V^\perp$:
    $\ket{w_{12}} \coloneqq \frac{1}{\sqrt6}(1,1,-2)^T$,
    $\ket{w_{13}} \coloneqq \frac{1}{\sqrt6}(1,-2,1)^T$, and
    $\ket{w_{23}} \coloneqq \frac{1}{\sqrt6}(-2,1,1)^T$.
    We can decompose the 3-copy Fock space as
    \begin{align*}
        \mathcal F_\vee(\CC^m \ot \CC^3)
        &\cong \mathcal F_\vee(\CC^m \ot (V \op W))
        \cong \mathcal F_\vee(\CC^m \ot V) \ot \mathcal F_\vee(\CC^m \ot W) \\
        &\cong \mathcal F_\vee(\CC^m \ot V) \ot \bigoplus_{k=0}^\infty \Sym^k(\CC^m \ot W).
    \end{align*}
    Then, $Q_{ab}$ decomposes accordingly as
    \begin{align*}
        Q_{ab} \cong I \ot \bigoplus_{k=0}^\infty Q_{ab}^{(k)},
    \end{align*}
    where $Q_{ab}^{(k)}$ is the projection onto $\Sym^k(\CC^m) \ot \CC w_{ab}^{\ot k} \cong \Sym^k(\CC^m \ot \CC w_{ab}) \subseteq \Sym^k(\CC^m \ot W)$, and
    \begin{align*}
        S \cong I \ot \bigoplus_{k=0}^\infty S^{(k)}, \qquad S^{(k)} = \frac13\left(Q_{12}^{(k)} + Q_{13}^{(k)} + Q_{23}^{(k)}\right).
    \end{align*}
    The nonzero eigenvalues of $3S^{(k)}$ equal those of the Gram matrix of the three vectors $w_{ab}^{\ot k}$, that is,
    \begin{align*}
    G_k
    =\begin{pmatrix}
       1&s_k&s_k\\
       s_k&1&s_k\\
       s_k&s_k&1
     \end{pmatrix},
\end{align*}
where $s_k = \braket{w_{12}}{w_{13}}^k = \braket{w_{12}}{w_{23}}^k = \braket{w_{13}}{w_{23}}^k = \left(-\frac12 \right)^k$.
The spectrum of $G_k$ is then $1 + 2s_k$ and $1 - s_k$ with multiplicity two. For $k = 0$, $S^{(k)}$ has only 1 as an eigenvalue.
For $k = 1, 2$ the largest eigenvalue of $S^{(k)}$ is $\frac12$, and for larger values of $k$ it is smaller. This proves the result.
\end{proof}

Finally, we derive a robustness result for testing general pure Gaussian states.

\begin{theorem}\label{thm:main boson}
    Let $\psi \in \P(\HH_m)$. If $p(\psi) \geq 1 - \tau$, then $F^2_{\GG_m}(\psi) \geq 1 - 48 \tau$.
\end{theorem}

\begin{proof}
    Suppose that the test accepts with probability at least $1 - \tau$, and assume $\tau \leq \frac{1}{48}$ (otherwise the bound is trivial). If we let
    \begin{align*}
        \rho \coloneqq \tr_1[U_{\pi/4}^\dagger \psi^{\ot 2} U_{\pi/4}]
    \end{align*}
    be the reduced state on the second copy after applying $U_{\pi/4}^\dagger$, then by \cref{eq:four copy prob} we have $\tr[P_0 \rho^{\ot 2}] \geq 1 - \tau$.
    Let $\lambda$ be the largest eigenvalue of $\rho$, and let $\ket{\xi}$ be a corresponding eigenstate.
    By \cref{lem:purity},
    \begin{align}\label{eq:pretty pure}
      \lambda \geq 1 - 2\tau,
    \end{align}
    and we also have
    \begin{align*}
        \tau \geq \tr[(I - P_0) \rho^{\ot 2}] \geq \lambda^2 \tr[(I - P_0) \xi^{\ot 2}] = \lambda^2 \left( 1 - p_0(\xi) \right),
    \end{align*}
    so $p_0(\xi) \geq 1 - \tau / \lambda^2$.
    By \cref{thm:main boson mean zero}, this implies that there exists a mean-zero Gaussian state, i.e.\ $\ket\gamma \coloneqq U_S^\dagger \ket0$ for some symplectic Gaussian unitary~$U_S$, such that
    \begin{align}\label{eq:mean zero fidelity}
      \abs{\angles{0 | U_S | \xi}}^2 = \abs{\angles{\gamma|\xi}}^2 \geq 1 - 8\tau / \lambda^2.
    \end{align}

    We will show that $\ket\phi \coloneqq U_S \ket\psi$ is close to a coherent state.
    Let
    \begin{align*}
      \sigma
    \coloneqq \tr_1[U_{\pi/4}^\dagger \phi^{\ot 2} U_{\pi/4}]
    = U_S \tr_1[U_{\pi/4}^\dagger \psi^{\ot 2} U_{\pi/4}] U_S^\dagger
    = U_S \rho U_S^\dagger,
    \end{align*}
    where we used that $U_{\pi/4}$ commutes with $U_S^{\ot 2}$.
    Then, by \cref{eq:coherent prob}, the test $Q$ accepts $\phi^{\ot 2}$ with probability
    \begin{align*}
      q(\phi)
    = \ip{0 | \sigma | 0 }
    = \ip{0 | U_S \rho U_S^\dagger | 0 }
    \geq \lambda \ip{0 | U_S \xi U_S^\dagger | 0 }
    = \lambda \, \abs{\ip{0 | U_S | \xi}}^2
    \geq \lambda - \frac {8\tau} \lambda
    \geq 1 - 12 \tau,
    \end{align*}
    where the penultimate inequality uses \cref{eq:mean zero fidelity}, and the last inequality uses that $\lambda \geq 1 - 2\tau \geq \frac{23}{24}$ by \cref{eq:pretty pure} and our assumption that~$\tau \leq \frac1{48}$.
    Hence, $\phi$ is close to a coherent state: we have
    $F_{\mathcal C_m}(\phi)^2 \geq 1 - 48\tau$ by \cref{thm:main boson coherent}.
    We conclude that $\ket\psi = U_S^\dagger \ket\phi$ is equally close to a Gaussian state.
\end{proof}

The following corollary follows from \cref{thm:main boson mean zero,thm:main boson coherent,thm:main boson} by the same argument as \cref{cor:fermion-testing}.

\begin{corollary}\label{cor:boson-testing}
Let $\eps,\delta\in(0,1)$, and let $\GG$ be any of the following sets:
\begin{enumerate}[(i),itemsep=0pt]
    \item $\GG_m$, the set of bosonic Gaussian states,
    \item $\GG_m^0$, the set of bosonic mean-zero Gaussian states,
    \item $\mathcal C_m$, the set of bosonic coherent states,
\end{enumerate}
Then there exists a test that uses $\bigO(\frac1{\eps^2} \log \frac1\delta)$ copies of a pure state~$\psi$ to achieve the following:
if $\psi \in \GG$, the test always accepts, while if $F^2_{\GG}(\psi) \leq 1-\eps^2$, it rejects with probability at least~$1-\delta$.
The test for $\GG_m$ acts on four copies at a time; the tests for $\GG_m^0$ and $\mathcal C_m$ act on two copies at a time.
\end{corollary}

\section*{Acknowledgments}
\addcontentsline{toc}{section}{Acknowledgments}
We warmly thank David Gross for many discussions on Gaussian testing.
In particular, the idea of using the gradient flow to upper bound the distance of a state to the set of Gaussian states arose in discussions between him and the authors in summer 2024.

MW acknowledges support by the European Union (ERC Grant SYMOPTIC, 101040907), by the Deutsche Forschungsgemeinschaft (DFG, German Research Foundation, 556164098), by the Deutsche Forschungsgemeinschaft under Germany's Excellence Strategy~--~EXC-2111~--~390814868, by the German Federal Ministry of Research, Technology and Space (QuSol, 13N17173), and by the Klaus Tschira Foundation, and thanks Q-FARM and the Leinweber Institute for Theoretical Physics at Stanford University and the Simons Institute for the Theory of Computing at UC Berkeley for hospitality.

\subsection*{Use of AI}
The authors thought about the problem for quite some time, and eventually it was solved with important contributions by an LLM.
The human contributions were as follows:
we formulated the question, observed the relation to symmetries of copies of Gaussian states, conjectured that the tests analyzed in this work should have constant sample complexity, and realized the strategy of lower bounding the Gaussian fidelity by analyzing the gradient flow for the acceptance probability and lower bounding the gradient norm to control the total motion, inspired by \cite{kwok2018paulsen,gradflow}.
A key insight, however, has been provided by an LLM, namely that the gradient can be bounded in terms of the spectral gap of the three-copy operator~$S$, which in turn is bounded below by a spectral gap that has been computed in \cite{carlen2000many,diaconis2000bounds,maslen2003eigenvalues,carlen2003determination,caputo2008spectral} in a completely different context (the Kac walk~\cite{kac1956foundations}).
A final human contribution was to write an exposition of the results.

\bibliographystyle{unsrt}
\bibliography{library}

\end{document}